\documentclass[11pt]{article}
\usepackage{fullpage,amsmath,amsthm,amssymb,amscd,fancybox,epsfig,float,subfigure}
\usepackage{graphics}

\newtheorem{definition}{Definition}[section]

\newcommand{\beqn}{\begin{equation}}
\newcommand{\eeqn}{\end{equation}}

\newtheorem{theorem}{Theorem}

\def\reals{{{\rm l} \kern-.15em {\rm R}}}
\newcommand{\xymod}{|\xb-\yb|}
\newcommand{\xb}{{\bf x}}
\newcommand{\yb}{{\bf y}}

\newcommand{\uvec}{{\bf {u}}}

\newcommand{\nub}{\boldsymbol{\nu}}
\newcommand{\kappab}{\boldsymbol{\kappa}}

\newcommand{\sumj}{\sum_{j=1}^3}
\newcommand{\sumk}{\sum_{k=1}^3}

\newcommand{\xx}{\textbf{x}}             % bold x, not italic.   etc ....
\newcommand{\yy}{\textbf{y}}

\title{Simple and efficient representations for the fundamental solutions 
	of Stokes flow in a half-space}
\author{Z. Gimbutas
	\thanks{%
		Information Technology Laboratory,
		National Institute of Standards and Technology,
		325 Broadway, Mail Stop 891.01,
		Boulder, CO  80305-3328.
		{{\em email}: {\sf {zydrunas.gimbutas@nist.gov}}}. The work of this author was 
		supported in part by the 
		Office of the Assistant Secretary of Defense for Research and Engineering 
		and AFOSR under NSSEFF Program Award FA9550-10-1-0180.
		and in part by the National Science Foundation under
		grant DMS-0934733. 
		Contributions by staff of NIST, an agency of the U.S. Government, 
		are not subject to copyright within the United States.}
	\and L. Greengard
	\thanks{Simons Center for Data Analysis, Simons Foundation,
		160 Fifth Avenue, NY, NY 10010 and Courant Institute of Mathematical Sciences,
		New York University, 251 Mercer Street, New York, NY 10012-1110.
		{{\em email}: {\sf {greengard@cims.nyu.edu}}}. The work of this author was 
		supported in part by the 
		Office of the Assistant Secretary of Defense for Research and Engineering 
		and AFOSR under NSSEFF Program Award FA9550-10-1-0180,
		by the National Science Foundation under
		grant DMS-0934733, and 
		by the Applied Mathematical Sciences Program of the U.S. Department of Energy
		under Contract DEFGO288ER25053.}
	\and S. Veerapaneni
	\thanks{Department of Mathematics, University of Michigan, 530 Church Street, Ann Arbor, MI 48109.
		{{\em email}: {\sf {shravan@umich.edu}}}.
		The work of this author was supported by the NSF under grants DMS-1418964 and DMS-1224656.}
}
\begin{document}

\maketitle

\begin{abstract}
  We derive new formulas for the fundamental solutions of slow,
  viscous flow, governed by the Stokes equations, in a
  half-space. They are simpler than the classical representations
  obtained by Blake and collaborators, and can be efficiently
  implemented using existing fast solvers libraries. We show, for
  example, that the velocity field induced by a Stokeslet can be
  annihilated on the boundary (to establish a zero slip condition)
  using a single reflected Stokeslet combined with a single
  Papkovich-Neuber potential that involves only a scalar harmonic
  function. The new representation has a physically intuitive
  interpretation.
%\vskip .1in
%\newline\newline
%\noindent {\bf Keywords}: Stokes flow, Blake's formulas,
%Papkovich-Neuber solution

\end{abstract}

\section{Introduction}
Viscous flow of passive and active suspensions in the presence of an
infinite planar boundary is an important physical model in many areas
of science and engineering. It serves as a useful paradigm for
understanding the effect of confined geometries on the macroscopic
flow behavior of particulate flows, for example, that of bacterial
propulsion, cellular blood flow and colloidal suspensions,
\cite{KK,CJ,BB,ves3d,SL}.  In problems where the Reynolds number is
low, the ambient fluid is governed by the Stokes equations:
\begin{equation}
\mu\, \Delta \uvec(\xx) = \nabla p(\xx), 
\quad \nabla \cdot \uvec(\xx) = 0, \label{Stokes}
\end{equation}
where $\mu$ is the fluid viscosity, $\uvec(\xx) = (u_1(\xx), u_2(\xx),
u_3(\xx))$ is the velocity of the fluid, and $p(\xx)$ is the
pressure. Assuming the plane wall is located at $x_3 = 0$
and that the flow velocity decays in the far field, 
\[
\uvec(\xx)  \rightarrow 0 ,
\quad\text{as}\quad |\xx| \rightarrow \infty ,
\]
the no-slip boundary condition is
\begin{equation}
\uvec\big\vert_{x_3=0} = 0  .
\label{BCs}
\end{equation}

Boundary integral methods are particularly well-suited for
problems of this kind since they discretize the domain boundary alone
(resulting in many fewer degrees of freedom) and impose 
the decay and the no-slip conditions exactly.
Moreover, they avoid the need for artificial truncation
of the computational domain and can be solved rapidly and with 
great accuracy using fast algorithms such as the fast multipole
method (FMM) and high order accurate quadrature rules.

In order to reformulate the Stokes equations as a 
boundary integral equation, however, one needs to have
access to the Green's functions for the
half-space \cite{Lamb,HB}. A now classical approach to 
constructing this Green's function
is due to Blake and others
\cite{Blake1971, Blake1974a, Blake1974b,  AB2, AB1, Pozrikidis, Yu}. 
Unfortunately, the resulting formulas are rather
complicated, making them 
somewhat difficult to implement.

Here, we show that a much simpler alternative to the Blake solution
can be obtained by combining a free space image, which 
annihilates the {\em tangential} components of velocity, 
with a Papkovich-Neuber potential \cite{PN1,PN2}
which annihilates the normal component. 

The present paper is organized as follows. 
We discuss Papkovich-Neuber potentials and
the standard fundamental solutions for the Stokes
equations in section \ref{sc:prelim}.
(See, for example, \cite{Lamb, HB, KK, Pozrikidis}).
We also review Blake's formula for the Stokeslet in a half-space.
In section \ref{sc:pap}, we derive 
the image structures
for Stokeslets, stresslets, rotlets, and Stokes doublets. 

In the appendices, we provide the analogous formulas
for two-dimensional half-space Stokes kernels and for 
some problems of linear elasticity.

\section{Fundamental solutions, the Papkovich-Neuber representation,
and Blake's formulas}
\label{sc:prelim}

Before discussing the various standard fundamental solutions for the Stokes
equations in free-space, we introduce the Papkovich-Neuber 
representation, originally developed in \cite{PN1,PN2} for
problems of linear elasticity. Without loss of generality, we assume that the fluid viscosity $\mu = 1$ in the rest of the paper. 

\begin{definition}
Let $\phi(\xx), \xx \in \reals^3$ be a harmonic function. Then
the induced Papkovich-Neuber representation is defined to be the
paired vector field $\uvec$ and scalar field $p$ given by:
\begin{equation} \uvec(\xx) = x_3 \nabla \phi (\xx) - \left[ 
\begin{array}{c} 0 \\ 0 \\ \phi(\xx) \end{array} \right], 
\quad p(\xx) = 2 \frac{\partial \phi(\xx)}{\partial x_3}.
\label{PN}\end{equation}
\end{definition}

\noindent
It is straightforward to verify that $(\uvec,p)$
satisfy the Stokes equations \eqref{Stokes} (with $\mu = 1$). 

Suppose now that a force vector 
${\bf f} = (f_1, f_2, f_3)$ is applied to a viscous fluid
at a point $\yy$. Then, it is well-known that the induced velocity 
and pressure can be computed
using the Stokeslet (the single layer kernel):
\begin{equation} S_{ij}(\xx,\yy)= \frac{1}{8\pi} \left[ \frac{\delta_{ij}}{\xymod}+
 \frac{(x_i-y_i)(x_j-y_j)}{\xymod^3} \right],
\quad i,j=1,2,3,
\end{equation}
\begin{equation} 
P_{j}(\xx,\yy)= \frac{1}{4\pi} \frac{x_j-y_j}{\xymod^3},
\end{equation}
where $\delta_{ij}$ is the Kronecker
delta. More precisely, 
the velocity vector $\uvec(\xx)$ and pressure $p(\xx)$ are
given by 
\begin{equation}
u_i(\xx) = \sumj S_{ij}(\xx,\yy) f_j, \quad p(\xx) = \sumj P_{j}(\xx,\yy) f_j. 
\label{stokeslet}
\end{equation}

The stresslet (or double layer kernel) for the Stokes equations describes
the velocity induced by an infinitesimal displacement 
${\bf g} = (g_1, g_2, g_3) $ (sometimes called the
{\em double force} source strength) at a point $\yy$
with orientation vector $\nub = (\nu_1, \nu_2, \nu_3) $.
It is given by:
%\vspace{-.1in}
\begin{equation} T_{ijk} (\xx,\yy)= \frac{3}{4\pi} \frac{(x_i-y_i)(x_j-y_j)(x_k-y_k)}{|\xb-\yb|^5}, \quad i,j,k=1,2,3,
\end{equation}
\begin{equation} \Pi_{jk}(\xx,\yy)= \frac{1}{2\pi} \left[- \frac{\delta_{jk}}{\xymod^3} 
+ \frac{3(x_j-y_j)(x_k-y_k)}{\xymod^5} \right], 
\end{equation}
and the corresponding formulas for the velocity and
pressure at an arbitrary point $\xx$ are
\begin{equation}
u_i(\xx) = \sumj \sumk T_{ijk}(\xx,\yy) \nu_k g_j, \quad p(\xx) = 
\sumj \sumk \Pi_{jk}(\xx,\yy) \nu_k g_j.
\end{equation}

\subsection{Stokes flow in a half-space, Blake's formula}
Suppose now that a force vector ${\bf f}$ is applied to a viscous
fluid in the upper half-space ($x_3>0$). Then the corresponding 
Stokeslet-induced velocity field fails to satisfy the no-slip
condition (\ref{BCs}).
In order to annihilate the velocity field while satisfying the homogeneous
Stokes equations in the upper half-space, \cite{Blake1971} proposed 
the following image structure: 
\begin{equation}
S^{W}_{ij}(\xb,\yb) = S_{ij}(\xb,\yb) - S_{ij}(\xb,\yb^{I}) 
+ 2y_3^2 S_{ij}^{D}(\xb,\yb^{I}) - 2y_3 S_{ij}^{SD}(\xb,\yb^{I}), 
\label{blake1}
\end{equation}
where $\yy^{I}=(y_1,y_2,-y_3)$ is the reflected image location.
Here, $S^{D}$ is a {\em modified source doublet} given by 
\begin{equation}
S_{ij}^{D}(\xb,\yb) = \frac{1}{8\pi} (1-2\delta_{j3}) 
 \frac{\partial}{\partial x_j} 
\frac{x_i-y_i}{|\xb-\yb|^3},
\end{equation}
and $S^{SD}$ is a {\em modified Stokes doublet}, given by 
\begin{equation}
S_{ij}^{SD}(\xb,\yb) = (1-2\delta_{j3})
 \frac{\partial S_{i3}(\xb,\yb)}{\partial x_j}.
\end{equation}
Similar, but more involved, decompositions for 
the Stokes doublet and stresslet in a
half-space are given in \cite{Blake1974b} and
\cite{Pozrikidis}. 
Note that the computation of the modified source doublet $S_{ij}^D$
requires the evaluation of three distinct harmonic dipole fields.

\section{A new image formula}
\label{sc:pap}

In this section, we derive a simpler image structure, using the 
Papkovich-Neuber representation which involves only a single harmonic
function. $\yy^I$, as above, will denote the image location $(y_1,y_2,-y_3)$.
Now, however, we define reflected 
single force, double force,
and double force orientation vectors by negating their third components:
\begin{equation}
  {\bf f}^{I} = (f_1,f_2,-f_3), 
  \quad {\bf g}^{I} = (g_1,g_2,-g_3), 
  \quad \nub^{I} = (\nu_1,\nu_2,-\nu_3).
\end{equation}
We also make use of the harmonic potential due to 
a unit strength charge,
\begin{equation} G^{S} (\xx,\yy) = 
\frac{1}{4\pi |\xx-\yy|},
\end{equation}
the harmonic potential due to a unit strength dipole with 
orientation vector $\nub$,
\begin{equation} G^{D}[\nub] (\xx,\yy) = 
\sum_{i=1}^3 \nu_i \frac{\partial}{\partial y_i} \frac{1}{4\pi |\xx-\yy|},
\end{equation}
and the harmonic potential due to a unit strength quadrupole with
orientation vectors $\nub$ and $\kappab$,
\begin{equation} G^{Q}[\nub,\kappab] (\xx,\yy) = 
\sum_{i=1}^3 \sum_{j=1}^3 \nu_i \kappa_j \frac{\partial^2}{\partial y_i\partial 
y_j} \frac{1}{4\pi |\xx-\yy|}.
\end{equation}
Finally, recall that,  
through the Papkovich-Neuber representation (\ref{PN}),
the harmonic function $\phi(\xx)$ induces 
velocity and pressure fields that can be written in component
form as
\begin{equation}
 u_i(\xb) = 
  x_3 \frac{\partial}{\partial x_i} \phi(\xx) - \delta_{i3} \phi(\xx),
 \quad
 p(\xx) = 2 \frac{\partial}{\partial x_3} \phi(\xx) \, .
\end{equation}
Of particular note is the fact that at $x_3=0$, the only non-zero
velocity component is 
\begin{equation}
 u_3(\xx) = -\phi(\xx).
\end{equation}

This suggests a simple two-step strategy. First, annihilate the
tangential components of the velocity field induced by a Stokeslet
(single layer) or stresslet (double layer) kernel. It is easy to see
that this can be accomplished by subtracting the influence of 
a reflected single force ${\bf f}^{I} =
(f_1,f_2,-f_3)$, or double force ${\bf g}^{I} = (g_1,g_2,-g_3),
\nub^{I} = (\nu_1,\nu_2,-\nu_3)$ located at the image point $\yy^I$,
respectively.
It remains only to match the 
remaining non-zero normal component $u_3$, which we will do by a
judicious choice of the Papkovich-Neuber potential $\phi(\xx)$.

\subsection{The Stokeslet correction}

Following the discussion above, let us write out in more detail the
velocity field ${\bf v}$ remaining after subtracting the image
Stokeslet at $\yy^I$ from the original Stokeslet at $\yy$ (Fig.  1):
\begin{equation}
    v_i(\xx) = \sum_{j=1}^3 S_{ij} (\xx,\yy) f_j -
    \sum_{j=1}^3 S_{ij} (\xx,\yy^{I}) f^{I}_j.
\end{equation}
At the interface $x_3=0$, a straightforward computation yields
\begin{equation}
 v_1(\xx)=0, \quad v_2(\xx)=0, 
 \quad v_3(\xx)= -\frac{{f^{I}_3}}{4\pi} 
{\frac{1}{|\xb-\yb^{I}|} } 
-\frac{
   {y_3} }{4\pi} {\sum_{j=1}^3 \frac{x_j-y^{I}_j}{|\xb-\yb^{I}|^3}
   f^{I}_j}.
\label{vcomp}
\end{equation}

\begin{figure}
\begin{center}
\includegraphics[width=4in]{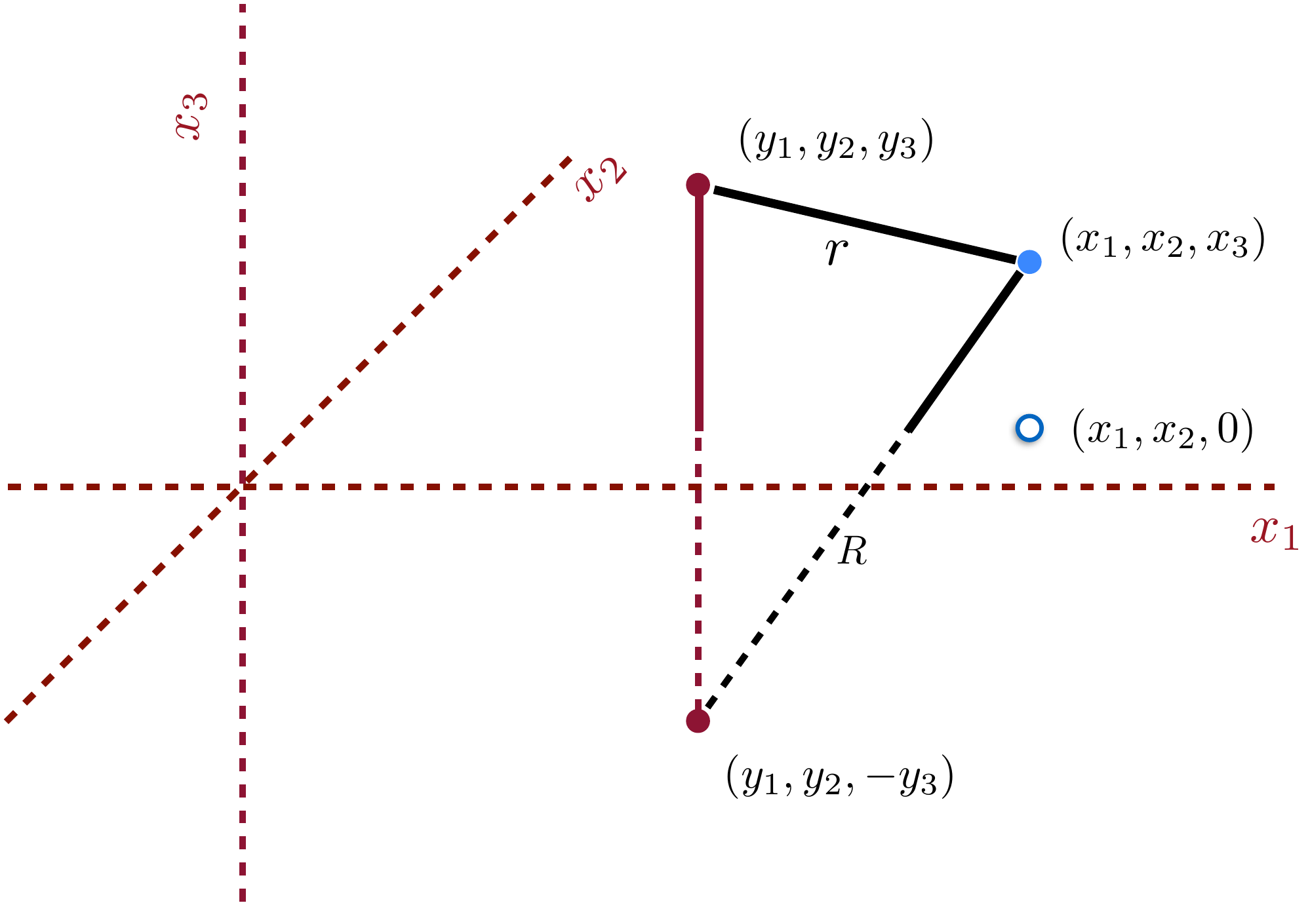}
\end{center}
\caption{The original source location is in the upper half-space
at $(y_1,y_2,y_3)$ and the reflected image source is at 
$(y_1,y_2,-y_3)$. For target points $(x_1,x_2,0)$ that lie
on the half-space boundary, the distances $r$ and $R$ from the 
source and image are the same, simplifying the computation of 
${\bf v}$ in (\ref{vcomp}).}
\end{figure}

By inspection of (\ref{vcomp}), it is clear that 
the harmonic potential $\phi$ required to cancel the non-zero normal
component $v_3$ at the interface is that induced by a
charge of strength $f_3^I$ and a dipole with orientation vector
${\bf f}^I$ of strength $y_3$ located at ${\bf y}^I$. 
Thus, the 
velocity field $\uvec^W$ satisfying the desired no-slip boundary
condition (\ref{BCs}) can be expressed as:
\begin{equation}
 u_i^{W}(\xx) =  u_i^{A}(\xx) -  u_i^{B}(\xx) -  u_i^{C}(\xx), 
\end{equation}
where $\uvec^A$ and $\uvec^B$ are the velocity fields induced by the 
original free-space Stokeslet and the reflected image force vector,
respectively:
\begin{equation}
 u_i^{A}(\xx) = \sum_{j=1}^3 S_{ij} (\xx,\yy) f_j, \qquad
 u_i^{B}(\xx) = \sum_{j=1}^3 S_{ij} (\xx,\yy^{I}) f^{I}_j.
\end{equation}
Here, $u_i^C$ is the Papkovich-Neuber correction
\begin{equation}
 u_i^{C}(\xx) = x_3 \frac{\partial}{\partial x_i} \phi(\xx) - \delta_{i3} 
\phi(\xx),
\end{equation}
where the harmonic potential $\phi$ is that due to a simple charge and dipole,
both located at $\yy^I$:
\begin{equation}
\phi(\xx) = f^{I}_3 \, G^{S}(\xx,\yy^{I}) + 
y_3 \, G^{D}[{\bf f}^{I}](\xx,\yy^{I}).
\end{equation}

\subsection{The stresslet correction}

We state the image structure for the stresslet (the double layer kernel)
in the form of a theorem.

\begin{theorem}
Let a double force ${\bf g}$ with orientation vector $\nub$ be located
at $\yy$, resulting in the free-space velocity field
\begin{equation}
 u^{T}_i(\xx) = \sumj \sumk T_{ijk} (\xx,\yy) \nu_k g_j \, .
\end{equation}
Then the corresponding velocity 
satisfying the no-slip boundary condition (\ref{BCs}) is given by 
\begin{equation}
 u^{T,W}_i(\xx) = 
\sumj \sumk T_{ijk} (\xx,\yy) \nu_k g_j - 
\sumj \sumk T_{ijk} (\xx,\yy^{I}) \nu^{I}_k g^{I}_j - 
\left[x_3 \frac{\partial}{\partial x_i} \phi^T(\xx) - 
\delta_{i3} \phi^T(\xx)\right] , 
\end{equation}
where
\begin{equation}
\phi^T(\xx) = 2 (\nub^{I} \cdot {\bf g}^{I})  G^{D}[{\bf h}](\xx,\yy^{I}) + 
2 y_3 G^{Q}[\nub^{I}, {\bf g}^{I}](\xx,\yy^{I}),
\end{equation}
with ${\bf h}=(0,0,1)$.
\end{theorem}

\begin{proof}
Note first that the 
symmetric part of the Stokes doublet is given by
\begin{equation}
   T^{S}_{ijk} (\xx,\yy) = \frac{1}{4\pi} \left[
    -\frac{(x_i-y_i)}{|\xb-\yb|^3} \delta_{jk} + 
      \frac{3 (x_i-y_i)(x_j-y_j)(x_k-y_k)}{|\xb-\yb|^5}
\right].
\end{equation}
Furthermore, the tangential components of the velocity field 
\begin{equation}
v_i(\xx) = \sumj \sumk T^{S}_{ijk} (\xx,\yy) \nu_k g_j - \sumj \sumk T^{S}_{ijk} (\xx,\yy^{I}) \nu^{I}_k g^{I}_j 
\end{equation}
are zero when $x_3=0$. 
Thus, the harmonic potential $\phi^{S}$ annihilating $v_3$ at
the interface $x_3=0$ is that due to a single quadrupole source 
located at ${\bf y}^I$ with orientation vectors $\nub^I, {\bf g}^I$ and strength $2 y_3$:
 \begin{equation}
 \phi^{S}(\xx) = 2 y_3 G^{Q}[\nub^{I}, {\bf g}^{I}](\xx,\yy^{I}). \label{pnsymm}
 \end{equation}
It is straightforward to check that
the stresslet is simply related to the symmetric part
of the Stokes doublet
\begin{equation}
   T_{ijk} (\xx,\yy) = \frac{3}{4\pi} 
      \frac{(x_i-y_i)(x_j-y_j)(x_k-y_k)}{|\xb-\yb|^5} 
=\frac{1}{4\pi} 
    \frac{x_i-y_i}{|\xb-\yb|^3} \delta_{jk} + T^{S}_{ijk} (\xx,\yy). 
    \label{double}
\end{equation}
Thus, it remains only to annihilate the 
velocity field induced 
by the first term on the right-hand side of (\ref{double})
and its reflected image, given by
\begin{equation}
v_i(\xx) = \frac{1}{4\pi} \sumj \sumk
 \frac{x_i-y_i}{|\xb-\yb|^3} \delta_{jk} \nu_k g_j - 
 \frac{1}{4\pi} \sumj \sumk \frac{x_i-y_i^I}{|\xb-\yb^I|^3} 
 \delta_{jk} \nu_k^I g_j^I \, .
\end{equation}
At $x_3=0$, the tangential components vanish and the normal
component is easily computed to be $2 (\nub^{I} \cdot {\bf g}^{I})
G^{D}[{\bf h}](\xx,\yy^{I})$, where ${\bf h} = (0,0,1)$.
The desired result follows. 
\end{proof}

\subsection{The rotlet correction}

Similar representations are easily derived for other fundamental
solutions, such as the {\em rotlet} - the antisymmetric part of the 
Stokes doublet:
\begin{equation}
   T^{R}_{ijk} (\xx,\yy) = \frac{1}{4\pi} \left[ 
 - \frac{(x_j-y_j)}{|\xb-\yb|^3} \delta_{ik} 
 + \frac{(x_k-y_k)}{|\xb-\yb|^3} \delta_{ij}
      \right].
\end{equation}

The corresponding velocity 
satisfying the no-slip boundary condition (\ref{BCs}) is given by 
\begin{equation}
 u^{R,W}_i(\xx) = 
\sumj \sumk T^R_{ijk} (\xx,\yy) \nu_k g_j - 
\sumj \sumk T^R_{ijk} (\xx,\yy^{I}) \nu^{I}_k g^{I}_j - 
\left[x_3 \frac{\partial}{\partial x_i} \phi^R(\xx) - 
\delta_{i3} \phi^R(\xx)\right] , 
\end{equation}
where the Papkovich-Neuber correction
$\phi^R$ is due to two dipoles:
\begin{equation}
\phi^R(\xx) = - 2 \nu^{I}_3 G^{D}[{\bf g}^{I}](\xx,\yy^{I}) + 
2 g^{I}_3 G^{D}[\nub^{I}](\xx,\yy^{I}). \label{pnanti}
\end{equation}

\subsection{The Stokes doublet correction}

Finally, the Stokes doublet is the sum of its symmetric and
antisymmetric parts
\begin{align}
   T^{D}_{ijk} (\xx,\yy) &= \frac{1}{4\pi} \left[ 
  - \frac{(x_i-y_i)}{|\xb-\yb|^3} \delta_{jk}
 - \frac{(x_j-y_j)}{|\xb-\yb|^3} \delta_{ik} 
 + \frac{(x_k-y_k)}{|\xb-\yb|^3} \delta_{ij} \right. \nonumber \\ 
  &+  \left.  \frac{3 (x_i-y_i)(x_j-y_j)(x_k-y_k)}{|\xb-\yb|^5} \right] 
= T^{S}_{ijk} (\xx,\yy) + T^{R}_{ijk} (\xx,\yy). 
\end{align}
By combining (\ref{pnsymm}) and (\ref{pnanti}), we obtain the
image formula for the Stokes doublet
\begin{align}
%\begin{equation}
 u^{D,W}_i(\xx) &= 
\sumj \sumk T^D_{ijk} (\xx,\yy) \nu_k g_j - 
\sumj \sumk T^D_{ijk} (\xx,\yy^{I}) \nu^{I}_k g^{I}_j \nonumber \\ 
&- \left[x_3 \frac{\partial}{\partial x_i} \phi^D(\xx) - 
\delta_{i3} \phi^D(\xx)\right] , 
%\end{equation}
\end{align}
where the Papkovich-Neuber potential
is given by 
\begin{equation}
\phi^D(\xx) = 
- 2 \nu^{I}_3 G^{D}[{\bf g}^{I}](\xx,\yy^{I}) 
+ 2 g^{I}_3 G^{D}[\nub^{I}](\xx,\yy^{I}) 
+ 2 y_3 G^{Q}[\nub^{I}, {\bf g}^{I}](\xx,\yy^{I}).
\end{equation}

% \section{An example}
% \label{sc:ex} 

% \begin{figure}[ht]
% \includegraphics[width=\textwidth]{evolution2}

% \caption{ \em \small Snapshots from simulation of a vesicle suspended
%   in a Stokesian fluid in a half-space and subjected to shear
%   flow. The governing equation for the vesicle membrane evolution is
%   of the form $\dot{\bf x} = \int_\gamma S({\bf x}, {\bf y}) {\bf
%     f}({\bf y}) \,d\gamma$ where ${\bf f}$ is the interfacial
%   force. More details on the formulation and the numerical method used
%   for this simulation can be found in \cite{ves3d}. The free-space
%   Green's function used in \cite{ves3d} is replaced with the
%   Papovich-Neuber formulae described in this paper. The wall-induced
%   diffusion of vesicle can clearly be seen here. }
% \label{fig:snaps1} \end{figure}

\section{Conclusions}
\label{sc:conclusions}
We have derived very simple image formulas for Stokes flow
in a half-space induced by any of the standard fundamental
solutions - the Stokeslet, stresslet, rotlet, and Stokes doublet.
In each case, all that is required is a reflected fundamental solution
and a Papkovich-Neuber correction based on a single harmonic
potential. 

The velocity (and pressure) due to the ``direct"
and reflected fundamental solutions can be computed together with
any software that handles Stokeslets, stresslets, etc. in free space.
Furthermore, the Papkovich-Neuber potential requires only the evaluation of 
a single additional harmonic function --- itself requiring only software
for free space {\em harmonic} sources, dipoles and quadrupoles.
Many efficient schemes exist for these various steps, such as those 
described in 
\cite{rodin1, rodin2, nishimura1, zorin1, frangi, 
wang2, duraiswami1, wangwhite, wang1, TG, fongdarve, CMCL}. \\

\appendix
\section{Extension to the two-dimensional problems}
The two-dimensional Stokes flow representations in a half-space can be
derived similarly.  They lead to identical Papkovich-Neuber
corrections with corresponding charge, dipole, and quadrupole
potentials replaced by their two-dimensional equivalents. Similar to
the three-dimensional case, for any ${\xx} \in \reals^2$, the
velocity field ${\uvec}^W$ satisfying the no-slip boundary condition in
a half-plane is composed of three terms:
\begin{equation}
 \uvec^W(\xx) =  \uvec^{A}(\xx) -  \uvec^{B}(\xx) -  \uvec^{C}(\xx),
\end{equation}
where the first term is the velocity field induced by the free-space
Green's function, the second term is the reflected image about the
plane wall annihilating the tangential velocity component, and the
third term is a Papkovich-Neuber correction term in the following
form:
\begin{equation} \uvec^C(\xx) = x_2 \nabla \phi (\xx) - \left[ 
\begin{array}{c} 0 \\ \phi(\xx) \end{array} \right], 
\quad p(\xx) = 2 \frac{\partial \phi(\xx)}{\partial x_2}. 
\label{PN2d}\end{equation}
The correction potentials for various fundamental solutions are
%%%
%%%\begin{equation} \arraycolsep=1.4pt\def\arraystretch{2.2}
%%%\begin{array}{ll}
%%%\phi(\xx) = \frac{f^{I}_2}{\mu} \, G^{S}(\xx,\yy^{I}) + \frac{y_2}{\mu} \, G^%{%D%}[f^{I}](\xx,\yy^{I})  & \quad \text{(stokeslet),} \\
%%%%
%%%   \phi^T(\xx) = 2 (\nu^{I} \cdot g^{I})  G^{D}[h](\xx,\yy^{I}) + 2 y_2 G^{Q}[\nu^{I}, g^{I}](\xx,\yy^{I}),  &  \quad \text{(stresslet),} \\
%%%%
%%% \phi^R(\xx) = - 2 \nu^{I}_2 G^{D}[g^{I}](\xx,\yy^{I}) + 
%%%2 g^{I}_2 G^{D}[\nu^{I}](\xx,\yy^{I})  &  \quad \text{(rotlet),} \\
%%%%
%%% \phi^D(\xx) =  - 2 \nu^{I}_2 G^{D}[g^{I}](\xx,\yy^{I}) 
%%%+ 2 g^{I}_2 G^{D}[\nu^{I}](\xx,\yy^{I}) 
%%%+ 2 y_2 G^{Q}[\nu^{I}, g^{I}](\xx,\yy^{I})  &\quad  \text{(doublet),} 
%%%\end{array}
%%%\end{equation}
%%
\begin{equation}
\phi(\xx) = f^{I}_2 \, G^{S}(\xx,\yy^{I}) + y_2 \, G^{D}[{\bf f}^{I}](\xx,\yy^{I}),
\end{equation}
for the two-dimensional Stokeslet, and
\begin{equation}
  \phi^T(\xx) = 2 (\nub^{I} \cdot {\bf g}^{I})  G^{D}[{\bf h}](\xx,\yy^{I}) + 2 y_2 G^{Q}[\nub^{I}, {\bf g}^{I}](\xx,\yy^{I}), 
\end{equation}
\begin{equation}
 \phi^R(\xx) = - 2 \nu^{I}_2 G^{D}[{\bf g}^{I}](\xx,\yy^{I}) + 
2 {\bf g}^{I}_2 G^{D}[\nub^{I}](\xx,\yy^{I}),
\end{equation}
\begin{equation}
 \phi^D(\xx) =  - 2 \nu^{I}_2 G^{D}[{\bf g}^{I}](\xx,\yy^{I}) 
+ 2 g^{I}_2 G^{D}[\nub^{I}](\xx,\yy^{I}) 
+ 2 y_2 G^{Q}[\nub^{I}, {\bf g}^{I}](\xx,\yy^{I}),
\end{equation}
for the two-dimensional stresslet, rotlet, and Stokes doublet, respectively,
where the images of the source, single force, double
force, and double force orientation vectors are
\begin{equation}
  \yy^{I} = (y_1, -y_2), \quad {\bf f}^{I} = (f_1, -f_2), 
  \quad {\bf g}^{I} = (g_1, -g_2), 
  \quad \nub^{I} = (\nu_1,-\nu_2),
\end{equation}
with the orientation vector ${\bf h}=(0,1)$, and the free-space Laplace Green's
functions in two dimensions are given by
\begin{equation} G^{S} (\xx,\yy) = - \frac{1}{2\pi} \log |\xx-\yy|, \quad
 G^{D}[\nub] (\xx,\yy) = 
 \sum_{i=1}^2 \nu_i \frac{\partial}{\partial y_i} G^{S} (\xx,\yy), \end{equation}
\begin{equation} 
\quad\text{and}\quad G^{Q}[\nub,\kappab] (\xx,\yy) = 
\sum_{i=1}^2 \sum_{j=1}^2 \nu_i \kappa_j \frac{\partial^2}{\partial y_i\partial 
y_j} G^{S} (\xx,\yy).
\end{equation}
\section{Extension to linear elasticity kernels}
The single layer kernel for linear isotropic elasticity in $R^3$ is given by
Kelvin's solution \cite{mindlin}
\begin{equation} U_{ij}(\xx,\yy)= \frac{1}{8\pi\mu} \left[ (2-\alpha) 
\frac{\delta_{ij}}{\xymod}+ \alpha \frac{(x_i-y_i)(x_j-y_j)}{\xymod^3} \right],
\quad i,j=1,2,3,
\end{equation}
where $\lambda$, $\mu$ are Lame's parameters and
$\alpha=(\lambda+\mu)/(\lambda+2\mu)$. It is easy to see that the displacement
field
\begin{equation}
  u_i(\xx) = \sumj S_{ij} (\xx,\yy) f_j - 
       \sumj S_{ij} (\xx,\yy^{I}) f^{I}_j - u_i^C(\xx)
\end{equation}
satisfies the no-displacement boundary condition $\uvec(\xx)=0$ at $x_3=0$,
where the Papkovich-Neuber correction is
\begin{equation}
  u^C_i(\xx) = \alpha
  x_3 \frac{\partial}{\partial x_i} 
  \phi(\xx) -(2-\alpha) \delta_{i3} \phi(\xx), 
\end{equation}
\begin{equation}
\phi(\xx) = \frac{f^{I}_3}{\mu} G^{S}(\xx,\yy^{I}) + 
\frac{y_3}{\mu} G^{D}[{\bf f}^{I}
](\xx,\yy^{I}).
\end{equation}

\vspace{.2in}

\bibliographystyle{plain}
\bibliography{stokeshf}

%\bibliography{stokeshfbib}
%\bibliographystyle{jfm}
\end{document}